\documentclass[showpacs,amsmath,amssymb,aps,prd,groupedaddress,superscriptaddress,nofootinbib]{revtex4-2}
\usepackage[utf8]{inputenc}
\usepackage{graphicx}
\usepackage{tensor}
\usepackage{bm}
\usepackage{color}

\newcommand{\leviconnection}[1]{\tensor{\mathring{\bm{\nabla}}}{#1}}%eu
\newcommand{\rsconnection}[1]{\tensor{\mathring{\omega}}{#1}}

\newcommand{\connection}[2][]{\tensor{#1{\Gamma}}{#2}}
\newcommand{\sconnection}[2][]{\tensor{#1\omega}{#2}}
\newcommand{\potential}[1]{\tensor{\Sigma}{#1}}%eu
\newcommand{\pd}[1]{\tensor{\partial}{#1}}%eu
\newcommand{\nablab}{\bm{\nabla}\!}%eu
\newcommand{\nablatensor}[2][]{\tensor{#1{\nablab}\!}{#2}}%eu
\newcommand{\e}[2][]{\tensor{#1{e}}{#2}}
\newcommand{\lorentz}[2][]{\tensor{#1\Lambda}{#2}}
\newcommand{\nonholon}[2][]{\tensor{#1{\Omega}}{#2}}

\newcommand{\bracket}[1]{\left<#1\right>}
\newcommand{\torsion}[1]{\tensor{T}{#1}}

\newcommand{\energy}[2][]{\tensor{{#1 t}}{#2}}

\newcommand{\bbtheta}{\hat{\theta}}
\newcommand{\bbthree}{\hat{\phi}}

\newcommand{\projection}[2][]{\tensor{{#1 b}}{#2}}

\newtheorem{theorem}{Theorem}[section]

\newcommand{\Weitzenbock}{Weitzenb{\" o}ck\ }
\newcommand{\volume}[1][]{{\cal V #1}}

\newcommand{\temperature}{\mathcal{T}}

\begin{document}

% Use the \preprint command to place your local institutional report
% number in the upper righthand corner of the title page in preprint mode.
% Multiple \preprint commands are allowed.
% Use the 'preprintnumbers' class option to override journal defaults
% to display numbers if necessary
%\preprint{}

%Title of paper
\title{Thermal equilibrium in general relativity and the special nature of the cosmological constant}

% repeat the \author .. \affiliation  etc. as needed
% \email, \thanks, \homepage, \altaffiliation all apply to the current
% author. Explanatory text should go in the []'s, actual e-mail
% address or url should go in the {}'s for \email and \homepage.
% Please use the appropriate macro foreach each type of information

% \affiliation command applies to all authors since the last
% \affiliation command. The \affiliation command should follow the
% other information
% \affiliation can be followed by \email, \homepage, \thanks as well.

\author{M. L. Rodrigues}
\email[]{marcoslucenarodrigues@hotmail.com}

\author{J. B. Formiga}
\email[]{jansen@fisica.ufpb.br}
%\homepage[]{Your web page}
%\thanks{}
%\altaffiliation{}

\affiliation{Departamento de Física, Universidade Federal da Paraíba, Caixa Postal 5008, 58051-970 João Pessoa, Pb, Brazil}

%\date{\today}

\begin{abstract}
The problem of thermal equilibrium in general relativity is discussed and the unique role played by the fluid associated with the cosmological constant is emphasized. In addition to being the only fluid that is always in thermal equilibrium, it is shown here that its energy can also be seen as a universal quantity given by $E_\Lambda=\rho_\Lambda\volume$, where $\volume$ is a spatial volume. Furthermore, if the cosmological constant turns out to be a universal constant, then the conclusion that any spatial volume $\volume$ has an intrinsic energy follows naturally from the approach considered here. A possible application to quantum theories of gravity that postulate the existence of a minimal length scale is discussed. In studying the energies, the so-called teleparallel equivalent of general relativity is used, and the difficulty with the zero-point energy is discussed in light of the gravitational energy problem. The quasi-local energies of the Reissner-Nordström-de Sitter black hole and of a regular black hole are calculated. The results obtained here with the teleparallel theory are consistent with our expectations.
\end{abstract}

% insert suggested PACS numbers in braces on next line
%\pacs{PACS numbers: 04.50.Kd, 02.40.Ma, 11.25.Hf, 04.30.-w}
% insert suggested keywords - APS authors don't need to do this
%\keywords{Teleparallel theories;.}

%\maketitle must follow title, authors, abstract, \pacs, and \keywords
\maketitle

\section{Introduction} % 1/1/2025

The concepts of thermodynamics have been extensively studied in general relativity, especially those regarding the laws of black-hole thermodynamics \cite{Bardeen1973,PhysRevD.7.2333,Hawking1975,PhysRevD.13.191,PhysRevD.14.870,PhysRevD.48.R3427} and fluid mechanics \cite{doi:10.1073/pnas.14.9.701,PhysRev.35.904,PhysRev.36.1791,PhysRev.58.919,PhysRev.76.427.2,0264-9381-10-8-017,doi:10.1142/S021827181846001X,PhysRevD.98.064001,Santiago_2019}. However, we do not yet have a clear understanding of some of those concepts. In this paper we use the standard notion of thermodynamic equilibrium to clarify some important features of general relativity and, in parallel, discuss the gravitation energy problem and the energy associated with the cosmological constant.

Some of the problems with thermodynamics in curved spacetimes are related to black holes and also with the notion of thermal equilibrium. Black hole thermodynamics is not as well established as we tend to think it is. There is no observation of Hawking radiation or any practical application of this thermodynamics. What we have is strong theoretical evidence that there is such a thing as a black hole thermodynamics. But even at the theoretical level we face some difficulties such as the correct form of the first law for regular black holes \cite{Ma_2014} and the definition of thermal equilibrium \cite{PhysRevD.110.104065}. Considering a perfect fluid, Tolman and Ehrenfest \cite{PhysRev.35.904,PhysRev.36.1791} established that a fluid at rest in a static spacetime is in thermal equilibrium if the product $T\sqrt{-g_{00}}$ is constant\footnote{For a pedagogical overview of Tolman temperature gradients, see Ref.~\cite{doi:10.1142/S021827181846001X}. For the generalization to stationary spacetimes, see Refs.~\cite{PhysRev.76.427.2} and \cite{PhysRevD.98.064001}.}, where $g_{00}$ is written in a coordinate system adapted to the time symmetry. This result can be seen as a particular case of Eckart's approach to fluid mechanics \cite{PhysRev.58.919}. In this approach, the energy-momentum tensor $\tensor{T}{^\mu^\nu}$ is decomposed into energy density $\rho$, isotropic pressure $p$, heat flux density $q^\mu$, and anisotropic stress tensor $\tensor{\pi}{^\mu^\nu}$. Then, thermal equilibrium is defined by $q^\mu=0$, and the Tolman-Ehrenfest criterion is recovered (see e.g. Appendix A of Ref.~\cite{PhysRevD.110.104065}). However, there is some controversy about whether thermal equilibrium can be established in some static spacetimes such as the Schwarzschild-de Sitter one. For example, in Ref.~\cite{PhysRevD.110.104065} the authors assume that the Schwarzschild-de Sitter horizons are at different temperatures and question whether it is possible to have $q^\mu=0$ as a condition for thermal equilibrium in this case. As a way out of this problem, they propose a generalization of thermal equilibrium for a static fluid in a static spacetime to a stationary heat conduction. In doing so, they neglect the backreaction of the fluid and show that, within this approximation, their definition implies that the two horizons remain in thermal equilibrium.

In this paper we will show some general results related to the condition $q^\mu=0$ and discuss the possible consequences. Section~\ref{20042026a} is devoted to the discussion of two results that have received little attention in the literature. The first concerns the fact that Einstein's field equations are incompatible with a nonvanishing heat flux from the viewpoint of static observers in a static spacetime, while the second shows the very special nature of the cosmological constant. We also argue that the two horizons of the Schwarzschild-de Sitter spacetime should be seen as having the same temperature, which answers the question raised in Ref.~\cite{PhysRevD.110.104065} without the need to change the notion of thermal equilibrium. In Sec.~\ref{20042026b}, we use the so-called teleparallel equivalent of general relativity (TEGR) to expand the discussion to the gravitational energy problem and show that the predictions of the TEGR are consistent with our expectations. As a byproduct, we show that the energy associated to the cosmological constant has a universal character and, in case $\Lambda$ is a universal constant, we can see this energy as a sort of ``geometrical energy'' associated to a given spatial volume. The consequences of this result are discussed and we argue that this energy may be useful for formulating a theory with a minimal length (or volume) without abandoning Lorentz invariance. We end Sec.~\ref{20042026b} by calculating the quasi-local energy of two interesting static spacetimes, the Reissner-Nordström-de Sitter black hole and a regular black hole. Finally, we summarize the conclusions in Sec.~\ref{20042026c}

\section{Thermal equilibrium}\label{20042026a}
Let $u^\mu$ be a four-velocity field. The projection tensor onto the three-space orthogonal to $u^\mu$ is
\begin{align}
\tensor{h}{_\mu_\nu}=g_{\mu\nu}+u_\mu u_\nu.
\label{04012025a}
\end{align}
It follows that the stress-energy tensor can be decomposed as
\begin{align}
\tensor{T}{^\mu^\nu}=\rho u^\mu u^\nu+p \tensor{h}{^\mu^\nu}+2q^{(\mu}u^{\nu)}+\tensor{\pi}{^\mu^\nu},
\label{14112025b}
\end{align}
where the energy density, the isotropic pressure, heat flux, and the anisotropic stress tensor are given by
\begin{align}
\rho=\tensor{T}{_\mu_\nu}u^\mu u^\nu,\ p=\frac{1}{3}\tensor{T}{_\mu_\nu}\tensor{h}{^\mu^\nu},
\end{align} 
\begin{align}
q^\mu=-T^{\alpha\beta}u_\alpha \tensor{h}{_\beta^\mu},\ \tensor{\pi}{^\mu^\nu}=\tensor{h}{^\mu_\alpha}\tensor{h}{^\nu_\beta}\tensor{T}{^\alpha^\beta}-p\tensor{h}{^\mu^\nu}.
\label{03012025c}
\end{align}
It is straightforward to verify that 
\begin{align}
\tensor{h}{_\mu_\nu}u^\nu=0,\ q_\mu u^\mu=0,\ \tensor{\pi}{_\mu_\nu}u^\nu=0,
\label{14112025a}
\end{align}
and both $\tensor{h}{_\mu_\nu}$ and $\tensor{\pi}{_\mu_\nu}$ are symmetric.

A natural definition of thermal equilibrium is $q^\mu=0$, which we assume throughout this paper.

\subsection{Heat flux in static spacetimes}\label{13112025a}
In this section we show a result that might already be known, but we have not found any proof in the literature. We state this result as a theorem.

\begin{theorem}\label{06022025a}
In a static spacetime, $q^\mu$ vanishes in a frame of reference adapted to a coordinate system $x^\mu$ where\footnote{Notice that $g^{0i}=0$ is not an addition restriction; it follows from $g_{0i}=0$ (see, e.g., the usual ADM decomposition \cite{Arnowitt2008}).}
\begin{align}
\partial_t g_{\mu\nu}=0,\quad g_{0i}=0,\quad (g^{0i}=0).
\label{03012025a}
\end{align}
\end{theorem}

\textit{Proof.} Einstein's field equations allow us to write $T^{\mu\nu}=2k G^{\mu\nu}$, where $T^{\mu\nu}$ is the energy-momentum tensor of the matter field, $G^{\mu\nu}$ is the Einstein's tensor, and $k=1/16\pi$. From this equation and Eq.~\eqref{03012025c} we see that the heat flux must satisfy  
\begin{align}
q^\mu=-2kG^{\alpha\beta}u_\alpha\tensor{h}{_\beta^\mu},
\label{09052026a}
\end{align}
where the components of the four-velocity of an observer adapted to the coordinates $x^\mu$ are
\begin{align}
u^\mu=\frac{\delta^\mu_0}{\sqrt{-g_{00}}},\quad u_\mu=-\sqrt{-g_{00}}\delta_\mu^0.
\label{03012025b}
\end{align}
From Eqs.~\eqref{04012025a}, \eqref{09052026a} and \eqref{03012025b}, we find that
\begin{align}
q^\mu=2k\sqrt{-g_{00}}(G^{0\mu}-G^{00}\delta^\mu_0).
\end{align} 
Since it is clear that $q^0=0$, we are left only with $q^i$. For $\mu=i$, we have $q^i=2k\sqrt{-g_{00}}G^{0i}$; since $g^{0i}=0$, we have $G^{0i}=R^{0i}$ and consequently  $q^i=2k\sqrt{-g_{00}}R^{0i}$. To show that this vanishes, we use $R^{0i}=g^{00}g^{ij}R_{0j}$ and
\begin{align}
R_{0j}=\partial_\alpha\connection{^\alpha_0_j}-\partial_j\connection{^\alpha_\alpha_0}+\connection{^\beta_0_j}\connection{^\alpha_\alpha_\beta}-\connection{^\beta_0_\alpha}\connection{^\alpha_j_\beta}.
\label{04012025b}
\end{align}
Using Eq.~\eqref{03012025a}, we obtain 
\begin{align}
\connection{^\alpha_0_j}=(1/2)g^{00}\partial_jg_{00}\delta^\alpha_0,\ \connection{^\alpha_0_0}=(-1/2)g^{\alpha i}\partial_ig_{00},\  \connection{^0_i_j}=0.
\end{align}
From these equations and Eq.~\eqref{03012025a}, we see that 
\begin{align}
\partial_\alpha\connection{^\alpha_0_j}=0,\ \connection{^\alpha_\alpha_0}=0,\ \connection{^\beta_0_j}\connection{^\alpha_\alpha_\beta}=0,\ \connection{^\beta_0_\alpha}\connection{^\alpha_j_\beta}=0.
\end{align}
Substitution into Eq.~\eqref{04012025b} yields $R_{0j}=0$, which leads to $R^{0i}=0$ and, therefore, $q^i=0$.

The only assumption we have made about the source in theorem \ref{06022025a} is that it satisfies Einstein's field equations. This means that any source that produces a static spacetime must yield $q^\mu=0$ in the reference frame where the coordinates are given by Eq.~\eqref{03012025a}. 

It is natural, therefore, to consider any spacetime with a timelike Killing vector as being in thermal equilibrium, or at least consider the fluid that is at rest in the coordinates adapted to the time symmetry to be in thermal equilibrium\footnote{Although we have not assumed that the fluid $\tensor{T}{_\mu_\nu}$ is at rest in the coordinates $x^\mu$, it seems clear that it has to be at rest because its motion would be incompatible with a static metric.}.

The case with a cosmological constant can be accounted for by taking $T_{\mathrm{\Lambda}}^{\mu\nu}=-2k\Lambda g^{\mu\nu}$ as part of the matter energy-momentum tensor. (For short, we call this fluid the $\Lambda$-fluid\footnote{It is generally believed to be related to a vacuum energy \cite{SeanCarrollSpacetimeandGeometry}}.) It turns out that both $q^\mu$ and $\tensor{\pi}{^\mu^\nu}$ vanish in this case for any $g_{\mu\nu}$ and any $u^\mu$, as will be shown in the next section.

According to Ref.~\cite{EllisBookRC}, when the strong energy condition is satisfied, there is a unique timelike $4$-velocity $u^\mu_{(e)}$ such that $q^\mu_{(e)}=-T^{\alpha\beta}u_{(e)\alpha } \tensor{h}{_{(e)}_\beta^\mu}=0$. The frame in which this happens is called \textit{energy frame} (or \textit{Landau-Lifshitz} frame). If such a frame always exist, whatever the spacetime, then thermal equilibrium is always possible.

\subsubsection{The Schwarzschild-de Sitter horizons}
In Ref.~\cite{PhysRevD.110.104065}, the authors raised the question of how the Schwarzschild–de Sitter horizons remain in thermal equilibrium at vastly different temperatures. They answered this question by changing the condition $q^\mu=0$ for thermal equilibrium to a nonvanishing $q^\mu$ that does not depend on time (stationary heat conduction). Here we propose a different explanation.

Instead of comparing two temperatures measured at two different locations, we focus on the temperature that an observer would attribute to photons coming from different locations when they arrive at the observer's location, assuming that the photon gas is in thermal equilibrium and obeys Planck's spectral law. This calculation is simple because of the ``law of energy red-shift'', which is given by $E_\mathrm{local}\sqrt{-g_{00}}=h\nu_0=$ constant (see e.g. Sec.~25.4 of Ref.~\cite{Gravitation}), where $h$ is Planck's constant. As already mentioned, the Tolman-Ehrenfest criterion is $T_\mathrm{local}\sqrt{-g_{00}}=T_0=$ constant. From these two equations, we see that an observer located at $r$ finds that $E_\mathrm{local}(r)/T_\mathrm{local}(r)=h\nu_0/T_0$. The dependency of Planck's law on temperature is exactly of the form $E_\mathrm{local}(r)/T_\mathrm{local}(r)$, which is constant. Therefore, regardless of the energies the photons had and the values of the local temperatures where they came from, the observer will perceive the whole universe as having only one temperature, the value $T_0$. (For a nice visualization of this thermal equilibrium, see Fig.~2 of Ref.~\cite{Santiago_2019}.

Of course, since $T_0$ is an integration constant, observers located at different values of $r$ will see the radiation in thermal equilibrium with different temperatures and, since there is no asymptotically flat region, we do not have a unique temperature to call the temperature of the horizons. Nonetheless, the conclusion that the horizons are at the same temperature is valid for all observers.

We therefore believe that there is no need to generalize the Tolman-Ehrenfest criterion for thermal equilibrium.

\subsection{The cosmological constant case}
The only fluid that is in thermal equilibrium in GR regardless of the spacetime geometry and observer's frame is a fluid of the type 
\begin{align}
\tensor{T}{^\mu^\nu}=-\rho_\Lambda g^{\mu\nu},
\label{28112025a}
\end{align}
where $\rho_\Lambda=2k\Lambda$ is constant ($\Lambda$ is the cosmological constant). To see this, we contract Eq.~\eqref{14112025b} with $u_\nu$, use Eq.~\eqref{14112025a}, and impose $q^\mu=0$. This gives
\begin{align}
\tensor{T}{^\mu^\nu}u_\nu=-\rho g^{\mu\nu}u_\nu,
\label{14112025c}
\end{align}
whose only solution for an arbitrary velocity is given by Eq.~\eqref{28112025a}.

Substitution of Eq.~\eqref{28112025a} into Eq.~\eqref{03012025c} yields $\tensor{\pi}{^\mu^\nu}=0$. Therefore, the fluid \eqref{28112025a} is the only perfect fluid that is ``perfect'' in all frames. 

Although the latter result is not new, see e.g. Sec. 5.2.2 of Ref.~\cite{EllisBookRC}, it is interesting to see it from the equations that relate two different decompositions. The general decomposition of an arbitrary $T^{\mu\nu}$ with respect to $u^\mu$ is given by Eq.~\eqref{14112025b}. Now, let $v^\mu$ be a different $4$-velocity vector field. The energy-momentum tensor $\tensor{T}{^\mu^\nu}$ can also be decomposed with respect to $v^\mu$:
\begin{align}
\tensor{T}{^\mu^\nu}=\rho_v v^\mu v^\nu+p_v \projection{^\mu^\nu}+2q_v^{(\mu}v^{\nu)}+\tensor{\pi}{_{v\!\!\!}^\mu^\nu},
\label{13032026a}
\\
\projection{_\mu_\nu}=g_{\mu\nu}+v_\mu v_\nu,\ \rho_v=\tensor{T}{_\mu_\nu}v^\mu v^\nu,\ p_v=\frac{1}{3}\tensor{T}{_\mu_\nu}\projection{^\mu^\nu},
\label{13032026b}
\\
q_v^\mu=-T^{\alpha\beta}v_\alpha \projection{_\beta^\mu},\ \tensor{\pi}{_{v\!\!\!}^\mu^\nu}=\projection{^\mu_\alpha}\projection{^\nu_\beta}\tensor{T}{^\alpha^\beta}-p_v\projection{^\mu^\nu},
\label{13032026c}
\end{align}
where 
\begin{align}
\projection{_\mu_\nu}v^\nu=0,\ q_{v\mu} v^\mu=0,\ \tensor{\pi}{_v_\mu_\nu}v^\nu=0,\ \tensor{\pi}{_{v\!\!}^\mu^\nu}\projection{_\mu_\nu}=0,
\label{13032026d}
\\
\tensor{\pi}{_{v\!\!}^\mu^\nu}g_{\mu\nu}=0.
\label{13032026e}
\end{align}
Substituting Eq.~\eqref{14112025b} into Eqs.~\eqref{13032026b} and \eqref{13032026c}, and using Eqs.~\eqref{13032026d}, \eqref{13032026e} and Eqs.~\eqref{02012026a}-\eqref{13032026f}, we get\footnote{These expressions can be obtained from Eqs.~(5.19)-(5.22) of Ref.~\cite{EllisBookRC} by changing $\tilde{u}^a$ to $v^\mu$, $\tilde{\rho}$ to $\rho_v$, $\tilde{q}^a$ to $q_v^\mu$, $\tilde{\pi}^{ab}$ to $\tensor{\pi}{_{v\!\!}^\mu^\nu}$, $v$ to $\beta$, and $v^a$ to $\gamma^{-1}h^\mu$. Note, however, that there is a difference here: the relative velocity we are using is between two observer at the same spacetime point, while in Ref.~\cite{EllisBookRC} the authors consider a finite displacement between them. In taking arbitrary values of the relative velocity, one has to assume that this displacement vanishes (see Sec.~4.5 of Ref.~\cite{EllisBookRC} for more details).}
\begin{align}
\rho_v=\gamma^2(\rho+p)-p-2\gamma q^\mu v_\mu+\tensor{\pi}{^\mu^\nu}v_\mu v_\nu,
\label{13032026h}
\end{align}
\begin{align}
p_v=\frac{1}{3}\left[\gamma^2(\rho+p)-\rho+2p-2\gamma q^\mu v_\mu+\tensor{\pi}{^\mu^\nu}v_\mu v_\nu\right],
\label{13032026i}
\end{align}
\begin{align}
q_v^\mu=\left[\gamma(\rho+p)-q^\lambda v_\lambda\right]b^\mu+\gamma q^\beta\projection{_\beta^\mu}-\tensor{\pi}{^\alpha^\beta}v_\alpha \projection{_\beta^\mu},
\label{13032026j}
\end{align}
\begin{align}
\tensor{\pi}{_{v\!\!}^\mu^\nu}=(\rho+p)[b^\mu b^\nu-\frac{1}{3}(\gamma^2-1)\projection{^\mu^\nu}]+2b^{(\mu}\projection{^{\nu)}_\lambda} q^\lambda
\nonumber\\
-\frac{1}{3}(-2\gamma q^\lambda v_\lambda+\tensor{\pi}{^\alpha^\beta}v_\alpha v_\beta)\projection{^\mu^\nu}+\tensor{\pi}{^\alpha^\beta}\projection{^\mu_\alpha}\projection{^\nu_\beta},
\label{13032026k}
\end{align}
where $\gamma$ is the Lorentz factor, and $b^\mu$ is the projection of $u^\mu$ onto the three-space orthogonal to $v^\mu$, given by Eq.~\eqref{02012026d} (see Appendix \ref{13032026g} for more details).

Equations \eqref{13032026h}-\eqref{13032026k}  are general, they hold for any energy-momentum tensor and for arbitrary timelike vector fields $u^\mu$ and $v^\mu$. Now when we consider the case of the cosmological constant, we have $p_\Lambda=-\rho_\Lambda$ and $q^\mu=\tensor{\pi}{^\mu^\nu}=0$. The above equations reduce to $\rho_v=\rho$, $p_v=-\rho$, and $q_v^\mu=\tensor{\pi}{_{v\!\!}^\mu^\nu}=0$.

\section{The gravitational case}\label{20042026b}
Despite the difficulties in defining an energy-momentum tensor for the gravitational field, and also a $4$-momentum for all spacetimes, it is interesting to check whether a promising approach to this problem can given consist results. Here, we test the spacetime $4$-momentum $P^a$ and the gravitational energy-momentum tensor $\energy{^\mu^a}$ of the TEGR. We show that they are consistent with the properties of a static fluid in a static spacetime by showing that the $3$-momentum $P^{(i)}$ and the heat fluxes vanish. We also show some interesting properties related to the $\Lambda$-fluid.

The choice of the TEGR is justified for the following reasons. In addition to being equivalent to general relativity, and therefore being in agreement with experiments, the TEGR deals with the gravitational energy problem in a more satisfactory manner: $\energy{^\mu^a}$ is traceless (compatible with a massless field) and independent of the coordinate system; the $4$-momentum $P^a$ can be seen as a generalization of the ADM momentum and has yielded many consistent results \cite{PhysRevD.106.044021}. 

Although the gravitation energy in the TEGR depends on the choice of what will be called here \textit{teleparallel frame} (TF),  and some choices lead to inconsistencies, it is possible to envisage a physical principle that guides this choice due to the natures of $P^a$ and $\energy{^\mu^a}$.

We give a brief review of the TEGR in the next section.

\subsection{Brief review}
In teleparallelism, the parallel transport of a vector is path-independent. This clearly includes special relativity as a particular case. However, when using the term ``teleparallelism'', one is usually referring to a theory with two affine connections: the Levi-Civita connection $\leviconnection{}$, which is usually represented by its coefficients in a coordinate basis, the Christoffel symbols, and a second connection $\nablatensor{}\ $ that is responsible for the distant parallelism.

There are many different formulations of teleparallelism \cite{doi:10.1142/S0219887823502158,AguiarGomes2023}, but we stick to the standard approach, whose connection $\nablatensor{}\ $ (called \Weitzenbock connection) has only torsion. In this case, there is always an orthonormal frame $\{\e{_a}\}=\{\e{_{(0)}},\e{_{(1)}},\e{_{(2)}},\e{_{(3)}}\}$ in which the connection coefficients of the \Weitzenbock connection vanishes\footnote{The coefficients of a connection in a tetrad basis are sometimes called \textit{spin connection}.}:
\begin{align}
\sconnection{^a_b_c}\equiv \bracket{\vartheta^a,\nablatensor{_b}\e{_c}}=0.
\label{31102025a}
\end{align}
(This is the definition of the TF.)

The torsion components can be written as $\torsion{^a_b_c}=2\sconnection{^a_{[bc]}}+\nonholon{^a_b_c}$, where $\nonholon{^a_b_c}$ is the object of anholonomity [see e.g. Eq.~\eqref{17032026a}], which  vanishes for coordinate bases\footnote{For example, in Minkowski, we have $\nonholon{^a_b_c}=0$ for $\e{_a}=(\partial_t,\partial_x,\partial_y,\partial_z)$.}. For a tetrad field satisfying Eq.~\eqref{31102025a}, we have $\torsion{^a_b_c}=\nonholon{^a_b_c}$. Hence, we can write $\torsion{^a_\mu_\nu}$ simply as
\begin{align}
\torsion{^a_\mu_\nu}=\pd{_\mu}\e{^a_\nu}-\pd{_\nu}\e{^a_\mu}.
\label{04102019p}
\end{align}

Einstein's field equations in the TEGR form is conveniently written as
\begin{equation}
\pd{_\alpha}\left( e\potential{^a^\mu^\alpha}\right)=\frac{e}{4k}\left(\energy{^\mu^a}+T^{\mu a}\right), \label{29032019k}
\end{equation}
where $\potential{_a_b_c}$ is called the \textit{superpotential}, $\energy{^\mu^a}$ is interpreted as the stress-energy tensor of the gravitational field, and $k=1/16\pi$ in geometrized units. In particular, one can use stokes' theorem to define a $4$-momentum for the spacetime:
\begin{align}
P^a=4k\oint_S dS_i e\Sigma^{a0i},
\label{30102025f}
\end{align}
where $S$ is the boundary of a region $V$. From Eq.~\eqref{29032019k} one can also write
\begin{align}
P^a=P_g^a+P^a_M,
\label{31032026b}
\end{align}
where
\begin{align}
P_g^a=\int_V d^3x e \energy{^0^a},\quad P_M^a=\int_V d^3x e T^{0a}
\label{24032026b}
\end{align}
are interpreted, respectively, as the gravitational and matter $4$-momenta within the region $V$.

Although $\potential{_a_b_c}$ and $\energy{^a^b}$ are given in terms of the torsion components \cite{ANDP:ANDP201200272}, when using the TF they can be written solely in terms of the Levi-Civita connection coefficients in the TF, denoted here by $\rsconnection{^a_b_c}$. In terms of this connection, we have (see, e.g., Sec.~5.1 of Ref.~\cite{doi:10.1142/S0219887824502116}): 
\begin{align}
\potential{_a_b_c}=\frac{1}{2}\rsconnection{_c_a_b}+\rsconnection{^d_d_{[c}}\tensor{\eta}{_{b]a}},
\label{30102025i}
\end{align}
\begin{align}
\energy{^b_a}= 2k\Bigl(2\rsconnection{^c_{[ad]}}\rsconnection{^b_c^d}-2\rsconnection{^b_{[ad]}}\rsconnection{^c_c^d}
-\rsconnection{^c_c_a}\rsconnection{^d_d^b}
\nonumber\\
+\delta^{b}_{a}\rsconnection{^c_{[c|f}}\rsconnection{^d_{|d]}^f}  \Bigr),
\label{30102025j}
\end{align}
where $\rsconnection{^a_b_c}$ can be given by
\begin{align}
\rsconnection{_a_b_c}=\frac{1}{2}(\torsion{_b_c_a}+\torsion{_c_b_a}-\torsion{_a_b_c}).
\label{30012025e}
\end{align}
[Equation \eqref{30012025e} is the well-known relation between $\rsconnection{_a_b_c}$ and $\nonholon{_a_b_c}$.]

Although $\energy{^b_a}$ is not necessarily symmetric, we apply the decomposition \eqref{14112025b} to its symmetric part. In fact, we are interested only in
\begin{align}
q^c_{\mathrm{g}}\equiv-\tensor{h}{_a^c}\energy{^{(ab)}}\eta_{(0)b}=-\frac{1}{2}\delta^c_i(\energy{^{(i)}_{(0)}}+\energy{_{(0)}^{(i)}}),
\label{30102025h}
\end{align}
where we set $\e{_{(0)}}=u$, which leads to $h_{ab}=\delta^i_a\delta^i_b=\delta^1_a\delta^1_b+\delta^2_a\delta^2_b+\delta^3_a\delta^3_b$ in the tetrad basis.

\subsection{Choosing the TF}
TEGR does not fix the kind of frame that satisfies Eq.~\eqref{31102025a}; we have to choose it. Since the solutions of Einstein's field equations are independent of this choice, one says that it is a gauge choice. However, when interpreting Eqs.~\eqref{30102025f} and \eqref{30102025j} as energies, we cannot treat Eq.~\eqref{31102025a} as a gauge: different choices lead to different energies and some choices are clearly inconsistent\footnote{This problem is similar to that of the pseudotensors of general relativity, but it is not exactly the same: the quantities defined by Eqs.~\eqref{30102025f} and \eqref{30102025j} have the same nature as the acceleration tensor, namely, they manifest the properties of the timelike congruence; it is even possible to calculate this tensor from  $\rsconnection{_a_b_c}$ alone.}. (For a discussion about possible principles to guide this choice, see Ref.~\cite{doi:10.1142/S0217732322502224}.)

We choose the TF to satisfy
\begin{align}
\e{_{(0)}^\mu}=u^\mu=\frac{\delta^\mu_0}{N},\quad N\equiv+\sqrt{-g_{00}}.
\label{29102025a}
\end{align}
This choice leads to
\begin{align}
\e{^{(i)}_0}=0.
\label{29102025b}
\end{align}
Furthermore, combining the restriction $g_{0i}=0$ of Eq.~\eqref{03012025a} with Eq.~\eqref{29102025a} leads to
\begin{align}
\e{^{(0)}_\mu}=N\delta_\mu^0,\ \e{_{(i)}^0}=0,
\label{29102025c}
\end{align}
which corresponds to the so-called \textit{time gauge}. In addition, we impose the condition
\begin{align}
\partial_0\e{^{(i)}_k}=0.
\label{30012025c}
\end{align}
[Note that the condition $\partial_0\e{^{(0)}_0}=0$ follows from Eqs.~\eqref{03012025a} and \eqref{29102025c}.]

The frame we have chosen is a natural choice from the physical point of view because it is at rest with respect to the fluid we want to study. In general, taking the frame as being at rest with respect to the system whose properties we are interested in serves as guidance for choosing the TF\footnote{This does not mean that we cannot study the physical system from the viewpoint of observers moving with respect to the fluid.}.

\subsection{Momentum and heat flux in the TEGR}
In thermal equilibrium, in addition to $q^\mu=0$, we may expect the $3$-momentum $P^{(j)}$ to vanish for any region with a boundary $S$. We now prove that this is the case for both matter and gravity.

In order to prove that $P^{(i)}=0$ and $q_{\mathrm{g}}^\mu=0$, we first show that 
\begin{align}
\rsconnection{_{(i)(0)(j)}}=\rsconnection{_{(i)(j)(0)}}=0.
\label{30012025d}
\end{align}
To find the values of $\rsconnection{_{(i)(0)(j)}}$ and $\rsconnection{_{(i)(j)(0)}}$, we need $\torsion{_{(0)(i)(j)}}$ and $\torsion{_{(i)(j)(0)}}$ only, since $\torsion{_a_b_c}=-\torsion{_a_c_b}$. From Eq.~\eqref{04102019p} we find
\begin{align}
\torsion{_{(0)(i)(j)}}=\e{_{(i)}^\mu}\e{_{(j)}^\nu}(\partial_\mu\e{_{(0)}_\nu}-\partial_\nu\e{_{(0)}_\mu}),
\label{30012025a}
\\
\torsion{_{(i)(j)(0)}}=\e{_{(j)}^\mu}\e{_{(0)}^\nu}(\partial_\mu\e{_{(i)}_\nu}-\partial_\nu\e{_{(i)}_\mu}).
\label{30012025b}
\end{align}
Applying the conditions in Eq.~\eqref{29102025c} to Eq.~\eqref{30012025a}, we see that $\torsion{_{(0)(i)(j)}}=0$. In turn, from Eqs.~\eqref{29102025a}-\eqref{30012025c}, we find $\torsion{_{(i)(j)(0)}}=0$. Now, using Eq.~\eqref{30012025e}, we finally see that Eq.~\eqref{30012025d} holds.

It is interesting to note that the result $\rsconnection{_{(i)(0)(j)}}=0$ implies that the rotations of the triad $\{\e{_{(i)}}\}$ vanish, while $\rsconnection{_{(i)(j)(0)}}=0$ implies that both the expansion and vorticity tensors vanish (See e.g. Sec.~II~D of Ref.~\cite{PhysRevD.108.044043}).

We see from Eq.~\eqref{30102025f} that we need to calculate $\potential{^{(j)}^0^i}$. Using Eq.~\eqref{30102025i} and imposing the conditions above, one finds that
\begin{align}
\potential{_{(j)}^0^i}=&\e{^{(0)0}}\e{^{(k)}^i}\potential{_{(j)(0)(k)}}
\nonumber\\
=&\frac{1}{2}\e{^{(0)0}}\e{^{(k)}^i}(\rsconnection{_{(k)(j)(0)}}+\rsconnection{^d_d_{(k)}}\eta_{(0)(j)}
-\rsconnection{^d_d_{(0)}}\eta_{(k)(j)})
\nonumber\\
=&0,
\label{30102025g}
\end{align}
where we have used Eq.~\eqref{30012025d}. (Notice that, since $\rsconnection{_a_b_c}=-\rsconnection{_c_b_a}$, we have $\rsconnection{^d_d_{(0)}}=\rsconnection{^{(i)}_{(i)}_{(0)}}=0$.) Hence, from Eq.~\eqref{30102025f}, we conclude that $P^{(j)}=0$.

Let us now focus on the heat fluxes. It is clear in Eq.~\eqref{30102025h} that we need to calculate only $q^{(i)}_{\mathrm{g}}$, because $q_{\mathrm{g}}^{(0)}=0$. Furthermore,  we have $\energy{^a_{(0)}}=\energy{_{(0)}^a}$. To see this, we use Eq.~(31) of Ref.~\cite{formiga2025angularmomentumteleparallelequivalent} and the conditions imposed on tetrad here. This gives
\begin{align}
2\energy{^{[(m)(0)]}}=-4k\left[\left(\partial_l\e{^{(0)}_0}\right)\potential{^{(m)}^l^0}-\left(\partial_l\e{^{(m)}_n}\right)\potential{^{(0)}^l^n}\right].
\end{align}
From the property $\potential{_{(m)}^0^l}=-\potential{_{(m)}^l^0}$ and Eq.~\eqref{30102025g}, we see that the first term on the right-hand side vanishes because 
$\potential{^{(m)}^l^0}=0$. In turn, the second term vanishes because $\potential{_{(0)}^l^n}=\e{^{(i)}^l}\e{^{(j)}^n}\potential{_{(0)(i)(j)}}=0$, where we have used Eqs.~\eqref{30102025i} and \eqref{30012025d}. Therefore $\energy{^a_{(0)}}=\energy{_{(0)}^a}$.

Now we are ready to prove that $q_{\mathrm{g}}^{(j)}$ vanishes. From Eq.~\eqref{30102025j} we obtain
\begin{align}
\energy{_{(j)(0)}}=&2k[\rsconnection{^c_{(0)}_d}\rsconnection{_{(j)}_c^d}-\rsconnection{^c_d_{(0)}}\rsconnection{_{(j)}_c^d}-\rsconnection{_{(j)(0)}_d}\rsconnection{^c_c^d}
\nonumber\\
&+\rsconnection{_{(j)}_d_{(0)}}\rsconnection{^c_c^d}-\rsconnection{^c_c_{(0)}}\rsconnection{^d_d_{(j)}}].
\label{30102025k}
\end{align}
It is straightforward to check from this equation and the results in Eq.~\eqref{30012025d} that  $\energy{_{(j)(0)}}=0$. Hence, we see from Eq.~\eqref{30102025h} that $q^c_{\mathrm{g}}=0$.

Since, in the TEGR, the total heat flux is given by\footnote{Notice that we have used $q^\mu$ in Sec.~\ref{20042026a} to denote the heat flow of the matter field, $T_{\mu\nu}$. Here, however, we use $q^\mu$ to represent the sum of the gravitation heat flow with that of the matter fields.} $q^c=q^c_{\mathrm{g}}+q_M^c$ and we have proved that $q^c_{\mathrm{g}}=q_M^c=0$, we conclude that there is no heat flux in a static spacetime from the viewpoint of a frame that satisfies the conditions given by Eqs.~\eqref{29102025a}, \eqref{29102025c} and \eqref{30012025c}.

From $\energy{_{(0)(j)}}=0$ and Eqs.~\eqref{24032026b} and \eqref{29102025c}, we see that $P_g^{(i)}=0$. 

In summary, we have $P_g^{(i)}=P_M^{(i)}=P^{(i)}=0$ and $q_g^\mu=q_M^\mu=q^\mu=0$.

\subsection{The $4$-momentum of the $\Lambda$-fluid}\label{06042026a}
Something very interesting happens when we assume that the $4$-momentum of the matter fields is given by Eq.~\eqref{24032026b}. Let us split $P^a_M$ into the $\Lambda$-fluid momentum and the other matter fields, i.e., $P^a_M=P_\Lambda^a+P^a_{M^\prime}$. From Eqs.~\eqref{28112025a} and \eqref{24032026b} we see that
\begin{align}
P_\Lambda^a=-\rho_\Lambda\int_V d^3x\, e\, \e{^a^0}.
\label{24032026c}
\end{align}
In the time gauge, i.e., when we chose $\e{^{(i)}^0}=0$, the three-momentum  $P_\Lambda^{(i)}$ vanishes.

Notice that this result does not depend on the choice of the TF, since Eq.~\eqref{24032026c} is independent of whether $\e{_a^\mu}$ is a TF: the choice of a TF affects only $P^a$ and $P_g^a$; the energy-momentum tensor $T_M^{\mu a}$ is proportional to the Einstein tensor $G^{\mu a}$, which does not depend on the choice of the TF.

Furthermore, the time gauge is not a restriction on the observers, it is just a relation between the hypersurface of constant $x^0$ and the frame (or coframe): in terms of the coframe $\vartheta^a=\e{^a_\mu} dx^\mu$, it means that $\vartheta^{(0)}\propto dx^0 $; in terms of the frame, it means that the triad $\e{_{(i)}}$ is inside the hypersurface. Therefore, the result $P_\Lambda^{(i)}=0$ can hold for any observer.

To interpret this result, let $\e{_a}$ be the rest frame of an observer $A$. If we assume that the components $\e{_a^\mu}$ in a coordinate system where the temporal coordinate is $x^0$ has to satisfy the time gauge, for consistency with the hypersurface of simultaneity $x^0= $~constant, then we arrive at the conclusion that the energy content of the fluid inside $x^0=$~constant is always at rest with respect to the observer $A$, whatever the observer's motion is.

Of course, if we apply a Lorentz boost, the new frame will not satisfy the time gauge in the old coordinate basis and the momentum inside the region $x^0=$ constant will not vanish. However, the new observer\footnote{The observer with respect to which the boosted frame is at rest.} predicts a vanishing three-momentum inside the hypersurface of simultaneity defined by a new time coordinate $\bar{x}^0$ that is supposed to be adapted to the new observer's rest frame.

To see this, consider two observers $A$ and $B$ whose rest frames are $\e{_a}$ and $\e[\bar]{_a}$, respectively. Assume that their frames are related to each other by means of the global Lorentz transformation $\e[\bar]{_a}=\lorentz{_a^b}\e{_b}$ with $\lorentz{_{(0)}^{(0)}}=\lorentz{_{(1)}^{(1)}}=\gamma$, $\lorentz{_{(0)}^{(1)}}=\lorentz{_{(1)}^{(0)}}=\beta\gamma$, and $\lorentz{_{(2)}^{(2)}}=\lorentz{_{(3)}^{(3)}}=1$, where $\gamma=1/\sqrt{1-\beta^2}$. Let us denote the coordinates adapted to $\e{_a}$ and $\e[\bar]{_a}$ by $x^\mu$ and $\bar{x}^\mu$, respectively. (It is assumed that $\e{_a}$ satisfies the time gauge in the basis $\{\partial/\partial x^\mu\}$, while $\e[\bar]{_a}$ satisfies this gauge in $\{\partial/\partial \bar{x}^\mu\}$.)  Applying Eq.~\eqref{24032026c} for both tetrads, we find
\begin{align}
E_\Lambda(\Sigma_A)=\rho_\Lambda \volume[\!\!]_A,\quad P_\Lambda^{(i)}(\Sigma_A)=0,
\label{26032026a}
\\
\bar{E}_\Lambda(\Sigma_B)=\rho_\Lambda \volume[\!]_B,\quad \bar{P}_\Lambda^{(i)}(\Sigma_B)=0,
\label{26032026b}
\end{align}
where $\Sigma_A$ and $\Sigma_B$ are the hypersurfaces of constant $x^0$ and $\bar{x}^0$, respectively, and the volumes $\volume[\!\!]_A$ and $\volume[\!]_B$ are given by\footnote{In the time gauge we have $e\e{_{(0)}^0}=\tensor[^3]{e}{}=\sqrt{\tensor[^3]{g}{}}$, where $\tensor[^3]{e}{}\equiv \det ||\e{^{(i)}_j}||$ and $\tensor[^3]{g}{}\equiv \det ||g_{ij}||$.}
\begin{align}
\volume[\!\!]_A=\int_{\Sigma_A}d^3x \sqrt{^3g},\quad \volume[\!]_B=\int_{\Sigma_B}d^3 \bar{x} \sqrt{^3\bar{g}},
\label{26032026c}
\end{align}
where $^3g$ and $^3\bar{g}$ are the determinants of $g_{ij}$ and $\bar{g}_{ij}$, the induced metrics on $\Sigma_A$ and $\Sigma_B$. (Both equations in Eq.~\eqref{26032026c} are invariant under coordinate transformations of the three-dimensional space.) Now, if we assume that the $4$-momentum of the fluid inside $\Sigma_A$ as seen by the observer $B$ is given by
\begin{align}
\bar{P}_\Lambda^a(\Sigma_A)=\int_{\Sigma_A} d^3x\, e\, \bar{T}^{0a},
\label{31032026a}
\end{align}
where the $0$ in $\bar{T}^{0a}$ is related to $x^0$ rather than $\bar{x}^0$, we find $\bar{P}_\Lambda^a(\Sigma_A)=\lorentz{^a_b}P_\Lambda^b(\Sigma_A)$, as expected. Therefore, $\bar{P}_\Lambda^{(1)}(\Sigma_A)=-\beta\gamma E_\Lambda$.

The latter calculation is exactly what is done in special relativity: one compares what $A$ and $B$ see in the same region of spacetime. In the limit case where the chosen\footnote{Although the integrals have been written in terms of the whole hypersurfaces, nothing prevents us from calculating the integrals \eqref{26032026c} and \eqref{31032026a} over some small regions inside them.} volume inside $\Sigma_A$ is small enough, we find the ordinary relation between the energy of a particle in its rest frame with that of a moving frame.

\subsubsection{Is $E_\Lambda$ universal?}
$\bar{P}_\Lambda^a(\Sigma_A)$ may seem more natural in the above sense. However, when we think in terms of the observers themselves, without worrying about their relation, what matters are the energies inside their own three-space, which are given by Eqs.~\eqref{26032026a} and \eqref{26032026b}. By shifting the focus from the energy inside the same region to the energy inside the same proper volume, i.e., $\volume[\!\!]_A=\volume[\!]_B$,  we conclude that all observers measure the same energy.

If we interpret the $\Lambda$-fluid as a fluid that is always present, either because it is the vacuum energy or because it is a sort of ``geometrical energy'', we could perhaps envisage a volume $\volume=\volume[\!\!]_A=\volume[\!]_B$ as having its own energy. For example, if we assume\footnote{Cosmological observations imply that $|\rho_\Lambda|\leq 10^{-8}\,$erg/cm$^3$ \cite{SeanCarrollSpacetimeandGeometry}. } that $\rho_\Lambda=10^{-8}\,$erg/cm$^3$, each cubic centimeter of space of any spacetime would intrinsically have $10^{-8}\,$erg of energy.

Notice that if $\rho_\Lambda$ is a universal constant, then the formula $E_\Lambda=\rho_\Lambda \volume$ is analogous to $E=mc^2$, i.e., spatial volume would be just a form of energy with $\rho_\Lambda$ playing the role of $c$. This analogy with $E=mc^2$ would become even more interesting if $E_\Lambda$ happens to be the smallest energy possible, in which case we could loosely say that it is the ``rest energy'' of the universe\footnote{Of course, to talk about the whole universe, we would have to consider $\volume$ as the spatial volume of the universe. To avoid a infinity volume, we could take $\volume$ as the volume of the observable universe.}; or more precisely, the minimal energy that is necessary to create a volume $\volume$.

\subsubsection{Minimal length scale and minimal value of $E_\Lambda$}
This interesting view of $\volume$ may be useful for theories of quantum gravity that claim that the spacetime has a minimal length scale. Those theories are commonly criticized for a possible conflict with Lorentz invariance when the same region of spacetime is compared by $A$ and $B$  \cite{PhysRevLett.88.190403}, i.e., when an equation like Eq.~\eqref{31032026a} is taken into account. However, it was argued in Ref.~\cite{PhysRevD.73.105013} that the minimal length should be measured in the rest frame of the two observers, which leads to a measurement of a different spacetime region for a boosted observer [like in Eqs.~\eqref{26032026a} and \eqref{26032026b}]; this is precisely the shift of focus described above. Hence, the universality of $\volume$ in the observers' rest frame may be suitable for the introduction of a minimal length.

To clarify this point, let us consider the following example. Assume that the minimal volume is the Planck volume, $\volume_p \approx 10^{-105}\,$m$^3$. Together with the assumption $\rho_\Lambda=10^{-8}\,$erg/cm$^3=10^{-9}\,$J/m$^3$ we obtain $E_{\Lambda,p}\approx 10^{-114}\,$J. Of course, this would not be the energy scale associated with the minimal length, because we would have to take into account the other energies in Eq.~\eqref{31032026b} that would be necessary to probe such small volume. However, since $E_{\Lambda,p}$ is fixed, it is independent of the processes used to measure the Planck volume. Therefore, through the relation $\volume_p=E_{\Lambda,p}/\rho_\Lambda$, one can establish the invariant notion of the scale\footnote{This does not mean that Planck energy is an invariant though. In principle, only the volume or length would be invariant.}.

Regardless of whether a minimal length contradicts Lorentz invariance, it is clear that a minimal energy for $E_\Lambda$ does not: in the context of Eq.~\eqref{31032026a}, we find $\bar{E}_\Lambda(\Sigma_A)=\gamma E_\Lambda(\Sigma_A) \geq E_\Lambda(\Sigma_A)$. Therefore, if $E_\Lambda$ has a minimal value in the rest frame of an observer, then no observer can measure a smaller value and all observers will agree on the minimal value. [It should also be clear from Eq.~\eqref{26032026a} that any minimal volume in the three-space of an observer implies a universal minimal value for $E_\Lambda$.]

In principle, relating the Planck energy $E_p$ to the total energy predicted by the TEGR, which is given by ``$E_{\mathrm{total}}=E_g+E_\Lambda+E_{M^\prime}$'' [see e.g. Eq.~\eqref{31032026b}], in an invariant way seems problematic. This is so because $E_{\mathrm{total}}$ depends on the choice of the TF even when the time gauge is assumed, and does not possess the nice property of depending solely on $\volume$. For example, the energy $E_{\mathrm{total}}$ vanishes for a flat universe \cite{Vargas2004}; this happens even at the level of the energy densities $\rho_g+\rho_\Lambda+\rho_{M^\prime}=0$ \cite{doi:10.1142/S021773232150125X}. The ``villain'' here is the gravitational energy density $\rho_g$, which is negative in this case. The TF of these results is adapted to freely falling particles\footnote{Notice that a freely falling frame is not necessarily a local inertial reference frame \cite{doi:10.1002/andp.201700175original}.} (a congruence of timelike geodesics) and satisfies the time gauge. On the other hand, it was shown in Ref.~\cite{doi:10.1142/S0217732322502224} that $\energy{^\mu^\nu}$ vanishes along the worldline of an observer when the TF is chosen to be adapted to the observer's proper reference frame (which also satisfies the time gauge); this means that $\rho_g=0$ along the observer's worldline and the total density reduces to the sum $\rho_\Lambda+\rho_{M^\prime}$, which does not vanish in general. (See Sec. 3.2 of Ref.~\cite{doi:10.1142/S0217732322502224}, especially the point 6 after theorem 2, for a discussion of this apparent contradiction.)

It seems more plausible that $E_p$ should be associated with the energy of the fields that can be used to probe distances, i.e., with $E_{M^\prime}$, despite the fact that this energy is not necessarily the same for the same volume of the observers' frames.

\subsubsection{Is $E_\Lambda$ the minimal energy of the universe?}

Based on the name ``vacuum energy'', which is commonly associated with the $\Lambda$-fluid, we are tempted to think that $E_\Lambda$ is the smallest energy a spacetime can have. However, the TEGR does not allow this conclusion trivially because of the ambiguity of $E_g$ described above ($E_g$ can even be negative, depending on the choice of the TF).

There are at least two cases in which the conclusion that $E_\Lambda$ is the minimal energy may hold. First, one could argue that the proper way to interpret Eq.~\eqref{31032026b} is by choosing a TF such that $E_g\geq 0$. The second possibility is to deny the role of $P^a$ as the physically meaningful spacetime $4$-momentum, due to the ambiguity of $P_g^a$, and use $P_{\mathrm{physical}}=P^a-P_g^a=P^a_\Lambda +P_{M^\prime}^a$ instead. In all these cases, $E_\Lambda$ would be the minimal energy possible because the total energy would be the sum of nonnegative energies.

However, there are some problems with those interpretations. The former interpretation seems too artificial, since there is no particular reason to think that the gravitation energy inside any region is always positive\footnote{In fact, in some situations, it is expected to be negative. See, for example, Sec.~19.8 of Ref.~\cite{PenroseRoadtoReality2007}.}. With respect to the latter, the $P_{\mathrm{physical}}$ defined above cannot reproduce the ADM $4$-momentum and predicts no energy for the gravitational field, including gravitational waves; this would be a complete denial of gravitational energy\footnote{See Ref.~\cite{AlemanBerenguer2025} for an interesting discussion about whether there is such a thing as a gravitational energy.}.

In any case, it is still possible to talk about $E_\Lambda$ as the vacuum energy of the matter fields, despite the above difficulties and the difficulties with the notion of vacuum in quantum field theory (QFT) \cite{SeanCarrollSpacetimeandGeometry}. In flat spacetimes, the zero-point energy is arbitrary, which means that the energy of vacuum is also arbitrary. However, in a curved spacetime, the energy of the so-called matter fields are well localized and therefore the minimal energy cannot be arbitrary. The ambiguity of the minimal energy discussed in this section is due to the gravitational energy $E_g$, which seems to be nonlocal; it is not because either $E_\Lambda$ or $E_{M^\prime}$ is nonlocal or arbitrary. So, it is possible to argue that $E_\Lambda$ is the minimal energy of the matter fields and, unlike what happens in QFT in flat spacetimes, this vacuum is universal and well defined.

Of course, the discussion here is classical and therefore cannot account for quantum effects. It is intended only to give a general perspective to the role that the $\Lambda$-fluid can play in gravity, with the hope that some of the conclusions can be carried to a quantum theory of gravity.

\subsection{Spacetime energy}
The results shown in Sec.~\ref{06042026a} hold for any spacetime. Let us now come back to the static case and obtain a simplified version of the total energy.

Contracting Eq.~\eqref{30102025i} with $\eta^{a(0)}\e{^b^0}\e{^c^j}$ and splitting the indices of $\rsconnection{^a_b_c}$ into spatial and temporal parts, we find, after some algebra, that
\begin{align}
\potential{^{(0)}^0^j}=\frac{1}{2}\left(\e{^{(0)0}}\e{^{(i)j}}-\e{^{(0)j}}\e{^{(i)0}}\right)\rsconnection{^{(k)}_{(k)}_{(i)}}
\nonumber\\
-\frac{1}{2}e{^{(k)0}}e{^{(i)j}}\rsconnection{_{(i)(0)(k)}}.
\label{08012026a}
\end{align}
The last term on the right depends on the rotation of the triad $\{\e{_{(i)}}\}$, since $\phi_{ab}=\rsconnection{_b_{(0)}_a}$ is the acceleration tensor.

Equation \eqref{08012026a} is a general expression for $\potential{^{(0)}^0^j}$. Imposing Eqs.~\eqref{29102025a}  and \eqref{29102025c}, we obtain the much simpler expression $\potential{^{(0)}^0^j}=-(1/2N)\e{^{(i)j}}\rsconnection{^{(k)}_{(k)(i)}}$; we also obtain the relation $e=N\, \tensor[^3]{e}{}$, where $\tensor[^3]{e}{}\equiv \det ||\e{^{(i)}_j}||$ is the determinant of the triad. Using these results in Eq.~\eqref{30102025f}, we find that\footnote{Equation \eqref{08012026b} is basically the same as  Eq.~(2.9) of Ref.~\cite{PhysRevD.65.124001}.}
\begin{align}
E=-2k\oint_S dS_j \tensor[^3]{e}{} \e{^{(i)j}}\rsconnection{^{(k)}_{(k)(i)}}.
\label{08012026b}
\end{align} 

Since the tetrad field does not depend on time, the spacetime energy does not depend either. Furthermore, $E$ depends only on the triad evaluated on the hypersurface of constant $x^0$. These properties are compatible with a static spacetime.

It is important to emphasize though that  Eq.~\eqref{08012026b} can still give inconsistent results. In general, it is necessary to avoid triads adapted to curvilinear coordinate systems.

\section{Static spherically symmetric spacetimes}
Let us decompose the components of the coframe in the form
\begin{align}
\e{^a_\mu}=-\hat{t}_\mu\hat{t}^a+\hat{r}_\mu\hat{r}^a+\hat{\theta}_\mu\hat{\theta}^a+\hat{\phi}_\mu\hat{\phi}^a,
\label{12012026e}
\end{align}
where the components of the unit vector fields $\{\hat{t},\hat{r},\hat{\theta},\hat{\phi}\}$ in the tetrad basis are defined by
\begin{align}
\hat{t}^a\equiv \delta^a_0,\ \hat{r}^a\equiv \sin\theta(\cos\phi\delta^a_1+\sin\phi\delta^a_2)+\cos\theta\delta^a_3,
\label{12012026f}
\\
\hat{\theta}^a\equiv\partial_\theta\hat{r}^a=\cos\theta(\cos\phi\delta^a_1+\sin\phi\delta^a_2)-\sin\theta\delta^a_3,
\label{12012026g}
\\
\hat{\phi}^a\equiv \partial_\phi(\hat{r}^a/\sin\theta)=-\sin\phi\delta^a_1+\cos\phi\delta^a_2.
\label{12012026h}
\end{align}
(For simplicity, we do not use the notation $\delta^a_{(0)}$, $\delta^a_{(1)}$,...,  since there is no risk of confusion when dealing with the Kronecker deltas: $\delta^{(1)}_1=\delta^1_1$, $\delta^{(1)}_0=\delta^1_0$ and so on.)

Now consider a general static spherically symmetric spacetime of the form
\begin{align}
ds^2=-B^2(r)dt^2+C^2(r)dr^2+r^2d\Omega^2.
\label{09012026a}
\end{align}
The above decomposition does not fix the tetrad field, and therefore the TF is still undefined. In choosing the TF, we take the components $\hat{t}_\mu$, $\hat{r}_\mu$, $\hat{\theta}_\mu$ and $\hat{\phi}_\mu$  as
\begin{align}
\hat{t}_\mu=-B(r)\delta^0_\mu,\ \hat{r}_\mu=C(r)\delta_\mu^1,\ \hat{\theta}_\mu=r\delta_\mu^2,\ \hat{\phi}_\mu=r\sin\theta\delta_\mu^3.
\label{12012026i}
\end{align}
We can use the inverse of the metric \eqref{09012026a} to raise the indices and obtain
\begin{align}
\hat{t}^{\mu}=\frac{1}{B(r)}\delta^\mu_0,\ \hat{r}^\mu=\frac{\delta^\mu_1}{C(r)},\ \hat{\theta}^\mu=\frac{1}{r}\delta^\mu_2,\ \hat{\phi}^\mu=\frac{\delta^\mu_3}{r\sin\theta}.
\label{18022026a}
\end{align}

The tetrad given by Eqs.~\eqref{12012026e}-\eqref{12012026h} and \eqref{12012026i} is adapted to isotropic rectangular coordinates\footnote{Although the form given by Eq.~\eqref{12012026d} is simpler than that of Eq.~\eqref{12012026e}, the latter is more convenient because, due to the spherical symmetry of the spacetime, it yields simpler expressions for quantities such as $\torsion{^a_b_c}$ and $\rsconnection{^a_b_c}$.} (Appendix \ref{12012026j}) and, therefore, it is not adapted to a curvilinear coordinate system.

The torsion components are (Appendix \ref{18022026c})
\begin{align}
\torsion{_a_b_c}=\frac{2B'}{BC}\hat{t}_a\hat{t}_{[b}\hat{r}_{c]}-\frac{2(C-1)}{rC}(\hat{\theta}_a\hat{r}_{[b}\hat{\theta}_{c]}+\hat{\phi}_a\hat{r}_{[b}\hat{\phi}_{c]}),
\label{18022026b}
\end{align}
where the prime denotes differentiation with respect to $r$. From now on, we deal only with algebraic calculations. From Eq.~\eqref{30012025e} and \eqref{18022026b}, we obtain
\begin{align}
\rsconnection{_a_b_c}=\frac{2B'}{BC}\hat{t}_b\hat{r}_{[a}\hat{t}_{c]}+\frac{2(C-1)}{rC}(\hat{\theta}_b\hat{r}_{[a}\hat{\theta}_{c]}+\hat{\phi}_b\hat{r}_{[a}\hat{\phi}_{c]}).
\label{18022026d}
\end{align}

The acceleration tensor $\tensor{\phi}{_a_b}$ can be computed from the relation $\tensor{\phi}{_a_b}=\rsconnection{_b_{(0)}_a}=\rsconnection{_b_c_a}\hat{t}^c$. Contracting Eq.~\eqref{18022026d} with $\hat{t}^b$, we find
\begin{align}
\tensor{\phi}{_a_b}=-\frac{2B'}{BC}\hat{t}_{[a}\hat{r}_{b]}.
\label{18022026e}
\end{align}
From this we see that $\tensor{\phi}{_{(i)(j)}}=0$, i.e. the frame is Fermi-Walker transported along the timelike geodesics $r$, $\theta$ and $\phi$ constant. (The triad does not rotate.) However, we have the acceleration $\tensor{\phi}{_{(0)(i)}}=(B'/BC)\hat{r}_{(i)}$. This means that the nongravitational forces exert no torque on the frame but they are accelerating it along $\hat{r}$.

From Eq.~\eqref{18022026d} we can also infer the energy given by Eq.~\eqref{08012026b}. Equation \eqref{18022026d} yields $\rsconnection{^{(k)}_{(k)(i)}}=-2(C-1)\hat{r}_{(i)}/rC$, while $\tensor[^3]{e}{}=\sqrt{\det ||g_{ij}||}$ gives $\tensor[^3]{e}{}=Cr^2\sin\theta$. Since $\e{^{(i)}^j}=\hat{r}^{(i)}\hat{r}^j+\hat{\theta}^{(i)}\hat{\theta}^j+\hat{\phi}^{(i)}\hat{\phi}^j$, we have $\tensor[^3]{e}{}\e{^{(i)}^j}\rsconnection{^{(k)}_{(k)(i)}}=-2(C-1)r\sin\theta\hat{r}^j=-2(1-1/C)r\sin\theta\delta_1^j$, where we have used Eq.~\eqref{18022026a} in the last equality. Thus, integrating  Eq.~\eqref{08012026b} over a sphere of radius $r$ leads to
\begin{align}
E=16\pi kr(1-C^{-1}).
\label{18022026f}
\end{align}

Notice that, when obtaining Eq.~\eqref{30102025f} [and of course Eq.~\eqref{08012026b}], we assumed that there is no singularity in the region inside $S$. This, however, does not exclude spacetimes with singularities, because we can enclose the singularities and take the spacetime energy as being the value of the integral in Eq.~\eqref{08012026b} on the external boundary.

\subsection{Worked examples}

\subsubsection{Reissner-Nordström-de Sitter black hole}
In this case we have $C(r)=1/B(r)$ and 
\begin{align}
B^2(r)=1-\frac{2M}{r}+\frac{Q^2}{r^2}-\frac{r^2}{L^2},
\label{24022026a}
\end{align}
where $M$ is the mass, $Q$ is the total electric charge, and $L$ is the de Sitter length\footnote{The cosmological constant is $\Lambda=3/L^2$.}. (When convenient, we can write $M=Gm/c^2$ and $Q=q\sqrt{G/4\pi\epsilon_0 c^4}$, where $G$ is the gravitational constant, $m$ the mass, $\epsilon_0$ the permittivity constant, $c$ the speed of light, and $q$ the charge in the International System of Units.)

Hence, from Eq.~\eqref{18022026f} the external boundary yields
\begin{align}
E=16\pi kr\left[1-\left(1-\frac{2M}{r}+\frac{Q^2}{r^2}-\frac{r^2}{L^2}\right)^{1/2}\right].
\label{24022026b}
\end{align}
This expression generalizes those of Refs.~\cite{PhysRevD.47.1407} and \cite{10.1063/1.531777}: the former, which was obtained from the gravitational action, does note have the cosmological term, while the latter does not have $Q$.

Assuming that there are at least two horizons at $r_\pm$ (vanishing $B$) for which $r_+>r_-$, the quasi-local energy inside the region $r_-\leq r\leq r_+$ is simply $16\pi k(r_+-r_-)$, which is positive.

\subsubsection{Regular black hole}

Let us apply Eq.~\eqref{18022026f} to a regular black hole in the context of the nonlinear electrodynamics (NED). We consider the case \cite{Kumar2021}
\begin{align}
B^2=1/C^2=1-\frac{2}{r}M e^{-\lambda^2/r},
\label{03032026a}
\end{align}
where $\lambda$ is a parameter that depends on the magnetic monopole charge. (This spacetime becomes the Schwarzschild spacetime when $\lambda=0$, and can be identified as the Reissner-Nordstr{\"o}m black for $r\gg \lambda^2$.) 

Substituting $C$ above into Eq.~\eqref{18022026f} we find
\begin{align}
E=16\pi k r\left[1-\left(1-\frac{2}{r}M e^{-\lambda^2/r}\right)^{1/2}\right].
\label{03032026b}
\end{align}
We obtain $E=M$ in the limit $r\to\infty$, and $0$ as $r \to 0^+$.

Depending on the value of $\lambda$, the term inside the parentheses in Eq.~\eqref{03032026b} can become negative. This does not happen when the maximum value of $(2M/r) e^{-\lambda^2/r}$ is less than unity; this corresponds to the condition $1>2Me^{-1}/\lambda^2$, in which case there is no black hole.

\section{Conclusions}\label{20042026c}
In this paper we have shown that, from the viewpoint of static observers in a static spacetime, the heat fluxes  must vanish when backreaction is not neglected. In addition,  we argued that the two horizons of the Schwarzschild-de Sitter spacetime can be seen as having the same temperature from a nonlocal point of view, and that this is the proper way to compare their temperatures when discussing thermal equilibrium. We have also emphasized the fact that the fluid associated to the cosmological constant, the $\Lambda$-fluid, is in thermal equilibrium in any spacetime, not only in static ones.

In order to investigate the energy of both the matter and gravitation fields, we have applied the TEGR approach to the gravitational energy problem using a frame that is adapted to static observers. (The rotations of the triad, the expansion tensor, and the vorticity of this frame vanish.) We have shown that all three-momenta vanish.

We have analyzed the energy and momentum of the $\Lambda$-fluid in an arbitrary spacetime and found that the three-momentum of this fluid vanishes in the time gauge. We have argued that the energy content of this fluid is at rest with respect to any observer. Furthermore, we have shown that we can see any spatial volume $\volume$ as having an intrinsic energy given by $\rho_\Lambda \volume$.  We have discussed the possibility that $E_\Lambda$ is the minimal energy and found that the TEGR does not allow that conclusion objectively, because of the zero-point energy of the gravitational field. Nevertheless, we can objectively see $E_\Lambda$ as the minimal energy of the matter fields, since the minimal energy of the other fields is zero (at least classically) and $E_\Lambda$ cannot vanish in a nonzero volume. Curiously, if the spacetime happens to have a minimal volume, as is proposed by some quantum theories of gravity, then $E_\Lambda$ will also have a minimal value. We have argued that this property may help to establish a theory with minimal length scale without breaking Lorentz invariance or resort to deformations of special relativity.

As an application of the TEGR approach, we calculated the total energy inside a region of radius $r$ for the cases of a Reissner-Nordström-de Sitter black hole and a regular black hole. In both cases, the results turned out to be consistent with our expectations.

% Specify following sections are appendices. Use \appendix* if there
% only one appendix.

\appendix

\section{Identities relating $u$ and $v$}\label{13032026g} % 20/11/2025
Let $u$ and $v$ be two timelike vector fields and $\{\e{_a}\}$ a frame such that $\e{_{(0)}}=u$. If we see $u$ as representing the $4$-velocity of an observer\footnote{In fact, a set of observers whose congruence are the integral curves of $u$.}, then $\{\e{_a}\}$ can be seen as the rest frame of this observer. Now let $v$ be the velocity field of a second observer. The spatial components of $v$ in $\{\e{_a}\}$, denoted by $v^{(i)}$, are proportional to the relative velocity of the second observer with respect to first one. To be more precise, the components $v^a$ are nothing but the Lorentz transformation $\lorentz{_{(0)}^a}$ from $\e{_a}$ to $v$: if $\{\e[\bar]{_a}\}$ is the rest frame of the second observer, where $\e[\bar]{_{(0)}}=v$, then $\e[\bar]{_a}=\lorentz{_a^b}\e{_b}$ implies that $v=\lorentz{_{(0)}^b}\e{_b}$; since $v=v^b\e{_b}$, we find that $v^b=\lorentz{_{(0)}^b}$.  Hence we have $v^{(0)}=\gamma$ and $v^{(i)}=\gamma\beta^{(i)}$, where $\beta^{(i)}$ is the velocity of the second observer with respect to the first one. In this case, the Lorentz factor can be written as $\gamma=1/\sqrt{1-\beta^2}$ with $\beta^2\equiv \beta^{(i)}\beta^{(j)}\eta_{(i)(j)}$. Therefore, we have $v\cdot u= v^au^b \eta_{ab}=v^a\delta^b_{(0)} \eta_{ab}=v^a\eta_{a(0)}=-v^{(0)}=-\gamma$.

From the above result, it is straightforward to show the identities:
\begin{align}
u\cdot v=-\gamma,
\label{02012026a}
\\
b_\nu\equiv u^\mu\projection{_\mu_\nu}=u_\nu-\gamma v_\nu,
\label{02012026d}
\\
h_\nu\equiv v^\mu\tensor{h}{_\mu_\nu}=v_\nu-\gamma u_\nu.
\label{02012026e}
\\
u^\mu b_\mu=u^\mu u^\nu \projection{_\mu_\nu}=v^\mu h_\mu=v^\mu v^\nu \tensor{h}{_\mu_\nu}=\gamma^2-1,
\label{02012026b}
\\
\tensor{h}{^\mu^\nu}\projection{_\mu_\nu}=\gamma^2+2,
\label{02012026c}
\\
\tensor{h}{^\alpha^\beta}\projection{^\mu_\alpha}\projection{^\nu_\beta}=\projection{^\mu^\nu}+b^\mu b^\nu,
\label{13032026f}
\\
\projection{^\alpha^\beta}\tensor{h}{^\mu_\alpha}\tensor{h}{^\nu_\beta}=\tensor{h}{^\mu^\nu}+h^\mu h^\nu,
\\
b_\mu\projection{^\mu^\nu}=b^\nu, \ h_\mu\tensor{h}{^\mu^\nu}=h^\nu.
\end{align}
Notice that $b^\mu$ is the projection of $u^\mu$ onto the three-space orthogonal to $v^\mu$, while $h^\mu$ is the projection of $v^\mu$ onto the three-space orthogonal to  $u^\mu$.

\section{Isotropic coordinate system}\label{12012026j}
The line element given by Eq.~\eqref{09012026a} can be put in the isotropic form
\begin{align}
ds^2=-B^2(r)dt^2+\frac{r^2}{\rho^2}(d\rho^2+\rho^2d\Omega^2)
\label{12012026a}
\end{align}
by means of the coordinate transformation
\begin{align}
\rho=\rho_0\exp\left(\int_{\rho_0}^r\frac{C(u)}{u}du\right).
\label{12012026b}
\end{align}
Using the coordinates $x$, $y$, and $z$ defined by
\begin{align}
x=\rho\sin\theta\cos\phi,\ y=\rho\sin\theta\sin\phi,\ z=\rho\cos\theta,
\label{12012026c}
\end{align}
we can write the tetrad that is adapted to the these coordinates as
\begin{align}
\bar{\vartheta}^{(0)}=Bdt,\ \bar{\vartheta}^{(1)}=\frac{r}{\rho}dx,\ \bar{\vartheta}^{(2)}=\frac{r}{\rho}dy,\ \bar{\vartheta}^{(3)}=\frac{r}{\rho}dz.
\label{12012026d}
\end{align}

To show that the frame given by Eqs.~\eqref{12012026e}-\eqref{12012026h} and \eqref{12012026i} is exactly the same as \eqref{12012026d}, one can simply evaluate $\vartheta^{(1)}$, $\vartheta^{(2)}$, and $\vartheta^{(3)}$ using these equations. For example,
\begin{align}
\vartheta^{(1)}=&\e{^{(1)}_\mu}dx^\mu=\hat{r}_\mu dx^\mu\hat{r}^{(1)}+\hat{\theta}_\mu dx^\mu\hat{\theta}^{(1)}+\hat{\phi}_\mu dx^\mu\hat{\phi}^{(1)}
\nonumber\\
=&C(r)dr\sin\theta\cos\phi+rd\theta\cos\theta\cos\phi-r\sin\theta\sin\phi d\phi
\nonumber\\
=&\frac{rd\rho}{\rho}\sin\theta\cos\phi+r\cos\theta\cos\phi d\theta-r\sin\theta\sin\phi d\phi
\nonumber\\
=&\frac{r}{\rho}(\sin\theta\cos\phi d\rho+\rho\cos\theta\cos\phi d\theta-\rho\sin\theta\sin\phi d\phi)
\nonumber\\
=&\frac{r}{\rho}dx,
\end{align}
where we have used $C(r)dr=(r/\rho)d\rho$, which can be obtained from Eq.~\eqref{12012026b}. An analogous calculation for $\vartheta^{(2)}$ and $\vartheta^{(3)}$  shows that $\vartheta^a$ is given by Eq.~\eqref{12012026d}. 

\section{Object of anholonomity}\label{18022026c}
From the definition of the Lie bracket of two vector fields $v$ and $w$ acting on a smooth function $f$, i.e. $[v,w]f\equiv v[w[f]]-w[v[f]]$, one can define the quantity
\begin{align}
\tensor{\Omega}{^a_b_c}\equiv-\bracket{\vartheta^a,[e_b,e_c]}=\e{_b^\mu}\e{_c^\nu}(\pd{_\mu}\e{^a_\nu}-\pd{_\nu}\e{^a_\mu}).
\label{17032026a}
\end{align}
This is called the object of anholonomity\footnote{One can find similar definitions in the literature such as ``the structure functions of the frame'' \cite{PhysRevD.91.084026} or commutation coefficients of the basis $\{\e{_a}\}$ \cite{Gravitation}. (Be careful with the sign convention.)}. It vanishes only when $\{\e{_a}\}$ is a coordinate basis.

From the orthonormality of the basis $\{\hat{t},\hat{r},\hat{\theta},\hat{\phi}\}$, it is easy to see that
\begin{align}
\nonholon{^a_\mu_\nu}=-\nonholon{^{\hat{t}}_\mu_\nu}\hat{t}^a+\nonholon{^{\hat{r}}_\mu_\nu}\hat{r}^a+\nonholon{^{\hat{\theta}}_\mu_\nu}\hat{\theta}^a+\nonholon{^{\hat{\phi}}_\mu_\nu}\hat{\phi}^a,
\end{align}
where $\nonholon{^{\hat{t}}_\mu_\nu}=\nonholon{^a_\mu_\nu}\hat{t}_a$, $\nonholon{^{\hat{r}}_\mu_\nu}=\nonholon{^a_\mu_\nu}\hat{r}_a$, and so on; they can be calculated from \cite{Formiga2021Braz} 
\begin{align}
\nonholon{^{\hat{t}}_\mu_\nu}=2\partial_{[\mu}\hat{t}_{\nu]}, \label{25112019c}
\\ 
\nonholon{^{\hat{r}}_\mu_\nu}=2\partial_{[\mu}\hat{r}_{\nu]}-2\delta^2_{[\mu}\hat{\theta}_{\nu]}-2\sin\theta\delta^3_{[\mu}\hat{\phi}_{\nu]}, \label{25112019d}
\\ 
\nonholon{^\bbtheta_\mu_\nu}=2\partial_{[\mu}\hat{\theta}_{\nu]}+2\delta^2_{[\mu}\hat{r}_{\nu]}-2\cos\theta\delta^3_{[\mu}\hat{\phi}_{\nu]}, \label{25112019dd}
\\ 
\nonholon{^\bbthree_\mu_\nu}=2\partial_{[\mu}\hat{\phi}_{\nu]}+2\sin\theta\delta^3_{[\mu}\hat{r}_{\nu]}+2\cos\theta\delta^3_{[\mu}\hat{\theta}_{\nu]}. \label{25112019e}
\end{align}

In a basis $\{\e{_a}\}$, the torsion components can be written as $\torsion{^a_b_c}=2\sconnection{^a_{[bc]}}+\nonholon{^a_b_c}$, where $\sconnection{^a_{bc}}$ are the connection coefficients expanded in this basis. Since the TF is defined by $\sconnection{^a_{bc}}=0$, the torsion components coincide with the anholonomity object in this frame. Hence, we can use the above equations to calculate $\torsion{^a_\mu_\nu}$. 

The viability of using this approach to calculate $\nonholon{^a_\mu_\nu}$, and consequently $\torsion{^a_\mu_\nu}$, depends on the symmetries involved. Here, the spherical symmetry combined with the Cartesian-like tetrad justify its use. From Eq.~\eqref{12012026i}, we see that $\delta^2_{[\mu}\hat{\theta}_{\nu]}=\delta^3_{[\mu}\hat{\phi}_{\nu]}=0$. In turn, when applying $\partial_\mu$ to $\hat{t}_\mu$, $\hat{r}_\mu$, $\hat{\theta}_\mu$ and $\hat{\phi}_\mu$, we can use the identities $\partial_\mu r=\delta^1_\mu$ and $\partial_\mu \theta=\delta^2_\mu$ and, then, we use Eq.~\eqref{12012026i} in order to write all the Kronecker deltas in terms of $\hat{t}_\mu$, $\hat{r}_\mu$, $\hat{\theta}_\mu$. Doing so leads to Eq.~\eqref{18022026b}.

% If you have acknowledgments, this puts in the proper section head.
\section*{Acknowledgments}
M. L. Rodrigues acknowledges CAPES for financial support.

%\bibliography{C:/Users/janse/OneDrive/Bibtex/bibJBF.bib}
%\bibliography{C:/Users/janse/OneDrive/Bibtex/bibJBF.bib}

%

\end{document}